\documentclass[conference]{IEEEtran}
\makeatletter
\def\ps@headings{%
\def\@oddhead{\mbox{}\scriptsize\rightmark \hfil \thepage}%
\def\@evenhead{\scriptsize\thepage \hfil \leftmark\mbox{}}%
\def\@oddfoot{}%
\def\@evenfoot{}}
\makeatother
\usepackage[T1]{fontenc}
\usepackage[latin9]{inputenc}
\usepackage{color}
\usepackage{float}
\usepackage{amsthm}
\usepackage{amsmath}
\usepackage{amssymb}
\usepackage{graphicx}
\usepackage{epstopdf}
\makeatletter

\floatstyle{ruled}
\newfloat{algorithm}{tbp}{loa}
\providecommand{\algorithmname}{Algorithm}
\floatname{algorithm}{\protect\algorithmname}

\theoremstyle{plain}

\theoremstyle{plain}

\theoremstyle{remark}

\newtheorem{proposition}{Proposition}

\usepackage{cite}

\makeatother

\providecommand{\corollaryname}{Corollary}
\providecommand{\remarkname}{Remark}
\providecommand{\theoremname}{Theorem}

\begin{document}

\title{Network Coding Techniques in Cooperative Cognitive Networks}

\author{\IEEEauthorblockN{Athanasios Papadopoulos\IEEEauthorrefmark{1}, Nestor
D. Chatzidiamantis\IEEEauthorrefmark{1} and Leonidas Georgiadis\IEEEauthorrefmark{1}}
\IEEEauthorblockA{\IEEEauthorrefmark{1}Department
of Electrical and Computer Engineering, Aristotle University of Thessaloniki,
Thessaloniki, Greece \\ Emails:\{athanapg,nestoras,leonid\}@auth.gr}
}
\maketitle

\begin{abstract}
In this paper, we investigate transmission techniques for  a fundamental cooperative cognitive radio network, i.e., a radio system where a Secondary user may act as relay for messages sent by the Primary user, hence offering  performance improvement of Primary user transmissions, while at the same time obtaining more transmission opportunities for its own transmissions. Specifically, we examine the possibility of improving the overall system performance by employing network coding techniques. The objective is to achieve this while affecting Primary user transmissions only positively, namely: 1) avoid network coding operations at the Primary transmitter in order avoid increase of its complexity and storage requirements, 2) keep the order of packets received by the Primary receiver the same as in the non cooperative case and 3) induce packet service times that are stochastically smaller than packet service times induced in the non-cooperative case. A network coding algorithm is investigated in terms of achieved throughput region and it is shown to  enlarge Secondary user throughput as compared to the case where the Secondary transmitter acts as a simple relay, while leaving the Primary user stability region unaffected. A notable feature of this algorithm is that it operates without knowledge of channel and packet arrival rate statistics. We also investigate a second network coding algorithm which increases the throughput region of the system; however, the latter algorithm requires knowledge of channel and packet arrival rate statistics.

\end{abstract}

\section{Introduction}
Cognitive radio networks (CRNs) received considerable attention due to their potential for improving spectral efficiency \cite{A:Haykin}. The main idea behind CRNs is to allow unlicensed users, known as Secondary users, to identify spatially or temporally available spectrum and gain access to the underutilized shared spectrum, while maintaining limited interference to the licensed user, also known as Primary user.

Initial designs of CRNs assumed that no interaction between Primary and Secondary users exists (see \cite{A:DSA1,DSA2} and the references therein). Of particular interest are the works of \cite{Peng, Chen-2008-ID52} which addressed the problem of optimal spectrum assignment to multiple Secondary users and presented resource allocation algorithms based on either the knowledge of Primary user transmissions obtained from perfect spectrum sensing mechanisms \cite{Peng} or from a probabilistic maximum collision constraint with the Primary Users \cite{Chen-2008-ID52}. Furthermore, in this framework, an opportunistic scheduling policy was suggested in \cite{Urgaonkar-2009-ID47}, which offered maximization of throughput utility for the Secondary users while providing guarantees on the number of collisions with the Primary user, as well.

By allowing cooperation between Primary and Secondary users in CRNs, cooperative CRNs have emerged. Cooperative CRNs have gained attention due to their potential of providing benefits for both types of users, i.e., by allowing Secondary users to relay Primary user transmissions, the channel between the Secondary transmitter and Primary receiver is exploited, improving the effective transmission rate of the Primary channel which as a result becomes idle more often, hence providing more transmission opportunities to the Secondary users.

Due to their advantages cooperative CRNs have been studied in several research works. From an information theoretic perspective, cooperation between Secondary and Primary users at the Physical layer has been investigated in \cite{Goldsmith-2009-ID359}. Of particular interest are the works which conduct queuing theoretic analysis and transmission protocol design for cooperative CRNs \cite{Simeone-2007-ID29,Krikidis-2009-ID233,Kompella,Neely_cooper, Nestor_ours,kiwan2017stability}. A cooperation transmission protocol for CRNs where the Secondary user acts as a relay for Primary user transmissions was initially presented in \cite{Simeone-2007-ID29}, where the benefits of such cooperation for both types of users were investigated. In \cite{Krikidis-2009-ID233}, cooperative CRNs with multiple Secondary users were investigated and advanced relaying techniques which involved advanced Physical layer coding between Primary and Secondary transmissions were suggested. Stationary transmission policies that allow simultaneous Primary and Secondary user transmissions were designed and optimized in \cite{Kompella}, in terms of stable throughput region. Cooperation transmission policies which take into account the available power resources at the Secondary transmitter in order for the latter to decide whether to cooperate or not, have been presented in \cite{Neely_cooper, Nestor_ours}; in these works cooperation between Primary and Secondary users is treated in an abstract manner (when cooperation takes place, transmission success probability is modified) without addressing in detail how this cooperation is effected. Finally, a cooperative CRNs with an extra dedicated relay was investigated in \cite{kiwan2017stability} and the maximum Secondary user throughput for this setup was determined. It should be noted that the implementation of all these  transmission algorithms for cooperative CRNs requires the modification of certain Primary channel parameters (such as Primary transmitted power, transmitted codewords, order of Primary transmitter packets received by Primary receiver), compared to the non-cooperative case, in order for the cooperation between Primary and Secondary users to take place and the true benefits of cooperation to be revealed.

In this work we examine the possibility of employing network coding techniques at the MAC layer to further improve the performance of cooperative cognitive networks. We impose the requirement that the Primary transmitter implementation complexity is minimally affected; no network coding operations are imposed on the Primary transmitter and the only additional requirement compared to the no cooperation case is that the Primary transmitter listens to Secondary transmitter feedback. Moreover, the order of Primary channel packet reception remains unaltered, while the service times of Primary packets are strictly improved compared to the case when no cooperation takes place. As in previous works, the Secondary transmitter may act as relay for Primary transmitter packets, however, depending on MAC channel feedback the Secondary transmitter may send network-coded packets that permit the simultaneous reception by the Primary and Secondary receivers. Under the aforementioned constraints on Primary channel transmissions, we propose an algorithm that implements network coding in this setup and study its performance in terms of Primary-Secondary throughput region. The results show that the resulting throughput region is improved as compared both to the cases where no cooperation takes place and the case where the Secondary transmitter acts as a relay but does not perform network coding. A notable feature of the algorithm is that the only requirement for its operation is knowledge that the channel from Secondary transmitter to Primary receiver is better than the channel from Primary transmitter to Primary receiver. Next, we examine the possibility whether it is possible to further increase the throughput region of the system by employing more sophisticated network coding techniques. We present a generalization of the proposed algorithm and show that this is possible in certain cases. However, in this case, knowledge of channel erasure probabilities, as well as the arrival rate of Primary transmitter packets are crucial for the algorithm to operate correctly.

The remainder of the paper is organized as follows. Section \ref{sec:System-Model} presents the system model along with some queuing theoretic preliminary results that are used in the analysis throughout the paper. In Section \ref{sec:baseline}, two baseline algorithms are described which will be used as benchmarks, when compared with the proposed one. Section \ref{subsec:Basic-Network-Coding} describes the proposed transmission algorithm that is based on network coding. In Section \ref{sec:Improved} we present a generalization of the algorithm proposed in Section \ref{subsec:Basic-Network-Coding} and show that it has increased throughput region in certain cases. Finally, concluding remarks are provided in Section \ref{sec:Conclusions}.

\section{System Model}
\label{sec:System-Model}
We consider the four-node cognitive radio system model depicted in
Fig. \ref{Fig:system_model}. The system consists of two (transmitter,
receiver) pairs (1,3), (2,4). Pair (1,3) - odd numbers- represents
the primary channel. Node 1 is the primary transmitter who is the
licensed owner of the channel and transmits whenever it has data to
send to primary receiver, node 3. On the other hand, node 2 is the
secondary transmitter; this node does not have any licensed spectrum
and seeks transmission opportunities on the primary channel in order
to deliver data to secondary receiver, node 4.
\begin{itemize}
\item \emph{Time and unit of transmission model. }We consider the time-slotted
model, where time slot $t=0,1,...$ corresponds to time interval $[t,t+1)$;
$t$ and $t+1$ are called the ``beginning'' and ``end'' of slot
$t$ respectively. Information transmission consists fixed size bits of packets
whose transmission takes unit time. At the beginning of time
slot $t$, a random number $A_{1}(t)$ of packets arrive at node 1
with destination node 3, thereafter called packets of session $(1,3).$
These packets are stored in an infinite-size queue $Q_{1}$. We assume
that the random variables $\left\{ A_{1}\left(t\right)\right\} _{t=0}^{\infty}$
are independent and identically distributed with mean $\lambda_{1}=\mathbb{E}\left[A_{1}(t)\right].$
Node 2 has an infinite number of packets destined to node 4, stored
in queue $Q_{2},$ thereafter called packets of session $(2,4)$.
The latter assumption amounts to assuming that node 2 is overloaded
and is made in order to simplify and clarify the presentation;
the algorithms presented still work and the results hold when
packet arrivals at node 2 are random.
\item \emph{Channel Model.} We consider the wireless broadcast channel,
i.e., that transmissions by node $i,\ i\in\left\{ 1,2\right\} $ may
be heard by the rest of the nodes. We adopt the broadcast erasure
channel model which efficiently describes communication at the MAC layer.
In this channel model, a transmission by node $i,\ i\in\left\{ 1,2\right\} $,
may either be received correctly by or erased at each of the other
nodes. Specifically, we make the following assumptions regarding the channel.
\begin{itemize}
\item \emph{Erasure events. }We assume that reception/erasure events are
independent across time slots, however, we allow for the possibility
that they be dependent within a given time slot. Specifically, For
a node $i\in\left\{ 1,2\right\} $, let $\left\{ Z_{k}^{i}(t)\right\}  _{t=1}^{\infty},k\in\left\{ 1,2,3,4\right\} $
be random variables, denoting erasure events, taking values 1 (a packet
transmitted by node $i$ is received by node $k$) and 0 (a packet
transmitted by node $i$ is erasure at node $i$). We define $Z_{i}^{i}\left(t\right)=1$
(a packet transmitted by node $i$ at time $t$ is known by node $i$)
and assume that the quadruples $\left(Z_{k}^{i}(t)\right)_{k=1}^{4},\ t=0,1,...$
are independent; however, for given $t,$ we allow for arbitrary dependence
between the random variables $Z_{k}^{i}\left(t\right),\ k\in\left\{ 1,2,3,4\right\}\ i\in {1,2} .$
We denote by $\epsilon_{{\mathcal S}}^{i},\ {\mathcal S}\subseteq\left\{ 1,2,3,4\right\} -\left\{ i\right\} $
the probability that a packet transmitted by node $i$ is erased at
all nodes in set ${\mathcal S}$. For simplicity we omit the brackets
when denoting specific sets in $\epsilon_{{\mathcal S}}^{i}.$ For example,
$\epsilon_{23}^{1}$ is the probability that a transmission by node
$1$ is erased at nodes $2,3$; the transmission may either be received
correctly or erased at node $4.$
\item \emph{Transmission scheduling}. We assume that simultaneous transmission
of packets by both transmitters results in loss of both packets; hence,
for useful transfer of information, only one of the transmitters must
be scheduled to transmit at any given time.
\item \emph{Channel feedback. }Upon reception or erasure of a packet, a
node sends respectively positive (ACK) or negative (NACK) acknowledgment
on a separate channel, which is heard by the rest of the nodes.
\item \emph{Channel sensing. }We assume that the Secondary transmitter can sense whether the Primary transmitter is sending a packet on the channel.
\end{itemize}
\end{itemize}
A main requirement in this setup is that node 2 transmissions must
either have no negative effect, or effect positively node 1 transmissions.
In the simplest case this can be achieved if transmitter 2 sends data
to receiver 4 only when transmitter 1 is idle. In this case, nodes
1 and 3 are effectively unaware of transmissions that take place between the secondary pair (2,4). However, if the
erasure probability from node 2 to node 3 is smaller than the one
from node 1 to node 3, i.e., $\epsilon_{3}^{2}<\epsilon_{3}^{1}$,
the possibility arises for improving the performance of both the primary
and the secondary channel by cooperation. Specifically, node 2 may
store packets sent by node 1 and erased by node 3 and then act as
a relay to transfer these packets to node 3. Since $\epsilon_{3}^{2}<\epsilon_{3}^{1}$
, this transfer will take shorter time. As a result the throughput
and packet delays for session (1,3) will improve and at the same time,
as long as $\lambda_{1}$ is not very high, node 1 will be idle for
a longer time and the throughput of packets for session (2,4) will
also increase.

In this work we examine the possibility of improving further the throughput
of packets of session (2,4) by allowing network coded transmissions
by node 2. We propose a network coding based algorithm according to which node 2 may transmit appropriate
combinations of packets destined to nodes 3, 4 which result in increased
throughput of packets of session (2,4). However, since node 1 is the
owner of the communication channel, in order to ensure that session
(1,3) transmissions are only positively affected we impose the following
requirements on the design of coding algorithms.

\textbf{Algorithm Design Requirements }
\begin{enumerate}
\item \label{enu:No-coding}No coding operations takes place at transmitter
node 1. Node 1 transmits its own packets based on the feedback received
by nodes 2, 3, 4, but does not receive/process any of the packets
transmitted by node 2.
\item \label{enu:OrderOfTrans}The order of packet transmission of session
(1,3) must be the same as in the case where no cooperation takes place.
\item \label{enu:The-service-time}The service time of each packet of session
(1,3) (i.e. the time interval between the time the packet is at the head of the queue on node 1 and the time the packet is successfully
received by node 3) must be ``smaller'' than the service time this
packet would have if no cooperation takes place. Specifically,
we require that if $S_{k}^{nc}$ ($S_{k}^{c}$) are the service times
of the $k$th session (1,3) packet when no cooperation (cooperation)
takes place, then $S_{k}^{c}$ is stochastically smaller than $S_{k}^{nc}$, that is,
\[
\Pr\left(S_{k}^{c}\geq x\right)\leq\Pr\left(S_{k}^{nc}\geq x\right),\ \mbox{for all \ensuremath{x\in[0,\infty)}}.
\]
\end{enumerate}
An algorithm that satisfies all three requirements stated above will
be called ``admissible''.

\begin{figure}
\centering\includegraphics[scale=0.4]{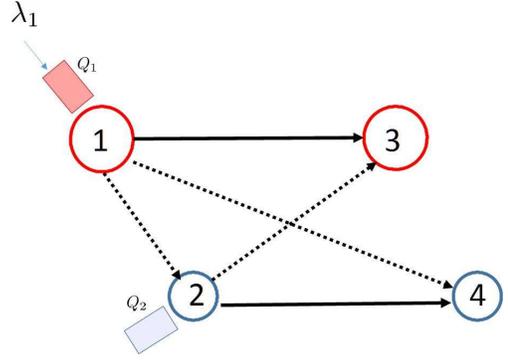}

\caption{The system model under consideration.}
\label{Fig:system_model}
\end{figure}

\subsection{Definitions and Preliminary Results}
\label{subsec:Prelim}

In the rest of this paper, for any storage element $X$ we denote
by $X(t)$ the number of packets in this element at time $t$.

A sequence of non-negative random variables $\left\{ Y\left(t\right)\right\} _{t=0}^{\infty}$
is stable if it converges in distribution to a proper random variable,
i.e,
\[
\lim_{t\rightarrow\infty}\Pr\left(Y(t)>M\right)=f(M)\ \mbox{for all \ensuremath{M\geq0},}
\]
and
\[
\lim_{M\rightarrow\infty}f\left(M\right)=0.
\]
One objective of the performance analysis of the algorithms to be
presented in the next sections is to determine the set of arrival
rates $\lambda_{1}$ for which the number of primary session (1,3)
packets \emph{in the system} at time $t$, denoted by $Q_{1}^{S}(t)$,
is stable. It will be seen that under all the algorithms discussed
in this paper, $Q_{1}^{S}\left(t\right)=Q_{1}(t)+F_{2}(t)$ where
$F_{2}(t)$ is a random variable taking values in $\left\{ 0,1\right\} $
and denotes the number of session (1,3) packets that may be located
at node 2. Also, under all algorithms $Q_{1}^{S}(t)$ can be seen
as the queue size of a discrete time queue where packets have independent identically distributed (i.i.d.)
service times with general distribution with mean $\bar{S}_{1}=1/\mu_{1}$.
Discrete time queues of this type have been studied in \cite{bruneel1993performance}
where it is shown that $Q_{1}^{S}\left(t\right)$ is stable when
\begin{equation}
0\leq\lambda_{1}<\mu_{1}.\label{eq:stabcond}
\end{equation}
Moreover, the average length of the busy and idle periods of $Q_{1}^{S}(t)$
are given respectively by,
\begin{eqnarray}
B_{1} & = & \frac{\lambda_{1}/\mu_{1}}{\left(1-\lambda_{1}/\mu_{1}\right)(1-q_{0})},\label{eq:BusyPeriod}\\
I_{1} & = & \frac{1}{(1-q_{0})},\label{eq:IdlePeriod}
\end{eqnarray}
where $q_{0}=1-\Pr\left(A_{1}(t)=0\right)$.

Let $R_{i}(t),\ i=1,2$ be the number of packets of session ($i,i+2$)
received by node $i+2$ during time slot $t$. The throughput $r_{i}$
of session ($i,i+2),\ i=1,2$ is defined as
\begin{equation}
r_{i}=\lim_{t\rightarrow\infty}\frac{1}{t}\sum_{\tau=0}^{t-1}R_{i}(t).\label{eq:thrptDef}
\end{equation}
It will be seen that for the algorithms discussed in this paper the
limit in (\ref{eq:thrptDef}) exists.

The objective of the algorithms presented in the next section is to
evaluate the maximum rate $r_{2}$ of session (2,4) packets that can
be obtained for given $\lambda_{1}$ satisfying condition (\ref{eq:stabcond});
under the latter condition, it is well known that it holds, $\lambda_{1}=r_{1}.$
The closure of the set of pairs $\left(r_{1},r_{2}\right)$ that can
be obtained under an algorithm is called ``throughput region'' of
the algorithm and is denoted by ${\mathcal R}.$

Next we present a \emph{generic queueing system} that will be used
for the performance analysis of the algorithms to be described
in the next sections. Consider a slotted time system with the following
structure. There are random time instants $T_{k},\ k=1,2,$ forming
a renewal process, i.e., $G_{k}=T_{k+1}-T_{k}\geq1$ are i. i. d.
with finite expectation. A random number $H_{k}$ of the slots in
the time interval $[T_{k},T_{k+1})$ are available for transmitting
the packets that are in the queue when this interval starts; the rest
of the slots are not available. Also, in the time interval $[T_{k},T_{k+1})$
a random number $A_{k}$ of packets arrives at the queue at various times; these packets are stored in an infinite size queue
and can be served at or after slot $T_{k}$. The $l$th arriving packet
needs $S_{l}$ of the available slots in order to be transmitted successfully.
The random variables $\left\{ A_{k}\right\} _{k=0}^{\infty},$$\ \left\{ H_{k}\right\} _{k=0}^{\infty},\ \left\{ S_{l}\right\} _{l=1}^{\infty}$
are i.i.d, and independent of each other, with finite expectations.
Let $r$ be the throughput of packets served by this queue. Using
arguments similar to those in \cite[Section 2]{neely2010stochastic}
it can be shown that
\begin{equation}
{\rm If}\ \mathbb{E}\left[A_{0}\right]\leq\frac{\mathbb{E}\left[H_{0}\right]}{\mathbb{E}\left[S_{1}\right]}\ {\rm then}\ r=\frac{\mathbb{E}\left[A_{0}\right]}{\mathbb{E}\left[G_{0}\right]}.\label{eq:stab2}
\end{equation}
\begin{equation}
{\rm If\ }\mathbb{E}\left[A_{0}\right]>\frac{\mathbb{E}\left[H_{0}\right]}{\mathbb{E}\left[S_{1}\right]},\ {\rm then\ }r=\frac{\mathbb{E}\left[H_{0}\right]}{\mathbb{E}\left[G_{0}\right]\mathbb{E}\left[S_{1}\right]}\label{eq:stab3}.
\end{equation}
A special case of this system is the discrete time queue in \cite{bruneel1993performance}
which is obtained by setting $T_{k}=k$,  $G_{k}=1$, and $H_{k}=1.$
\section{Baseline Algorithms}
\label{sec:baseline}
In this section we describe two baseline algorithms. The first involves no cooperation while in the second the secondary transmitter may be used as relay for session $(1,3)$ packets, but performs no network coding operations.

\subsection{No Cooperation }

Algorithm \ref{alg:NC-Alg} is very simple and requires no cooperation
between the Primary and Secondary users.

\begin{algorithm}[t]
\caption{\label{alg:NC-Alg}No cooperation Algorithm}

\begin{enumerate}
\item If $Q_{1}$ is nonempty, node 1 (re)transmits the packet at the head
of $Q_{1}$ until it is received by node 3.
\item \label{enu:AlgNC-st2}If $Q_{1}$ is empt,y node 2 (re)transmits the head of $Q_{2}$ packet until it is received by node 4.
\end{enumerate}
\end{algorithm}

To obtain the throughput region of Algorithm \ref{alg:NC-Alg}, observe
first that $Q_{1}$ is a discrete time queue in which the service
time of each packet is geometrically distributed with parameter $1-\epsilon_{3}^{1}$,
i.e., $\mu_{1}=1-\epsilon_{3}^{1}$. Hence, according to (\ref{eq:stabcond}),
this queue is stable when $0\leq\lambda_{1}<1-\epsilon_{3}^{1}$,
and according to (\ref{eq:BusyPeriod}), (\ref{eq:IdlePeriod}), the
average lengths of the busy and idle periods of $Q_{1}$ are given
respectively by,
\begin{eqnarray*}
B_{1} & = & \frac{\lambda_{1}/\left(1-\epsilon_{3}^{1}\right)}{\left(1-\lambda_{1}/\left(1-\epsilon_{3}^{1}\right)\right)(1-q_{0})},\\
I_{1} & = & \frac{1}{(1-q_{0})}.
\end{eqnarray*}
Since according to item \ref{enu:AlgNC-st2} of Algorithm \ref{alg:NC-Alg},
node 2 transmits session (2,4) packets during the idle periods of
queue $Q_{1},$ it can be shown based on arguments from regenerative
theory \cite{Wolff89} that
\begin{equation}
r_{2}=\frac{I_{1}\left(1-\epsilon_{4}^{2}\right)}{B_{1}+I_{1}}=\left(1-\frac{r_{1}}{1-\epsilon_{3}^{1}}\right)\left(1-\epsilon_{4}^{2}\right).\label{eq:NoCoopThroughput}
\end{equation}
Since any throughput for session (2,4) smaller that the one in (\ref{eq:NoCoopThroughput})
can also be achieved (the algorithm may simply refrain from transmitting
in certain slots), we see that the throughput region of Algorithm
\ref{eq:NoCoopThroughput} is
\[
{\mathcal R}_{\ref{alg:NC-Alg}}=\left\{ (r_{1},r_{2})\geq\boldsymbol{0}:\ \frac{r_{1}}{1-\epsilon_{3}^{1}}+\frac{r_{2}}{1-\epsilon_{4}^{2}}\leq1\right\} .
\]

\subsection{Simple Forwarding }

The algorithms presented in this and the following sections are admissible
when the channel from node 2 to node 3 is ``better'' that the channel
from node 1 to node 3. Specifically we assume for the rest of this work that
\begin{equation}
\epsilon_{3}^{1}\geq\epsilon_{3}^{2}.\label{eq:ErasureBasicAssumption}
\end{equation}
While the algorithms to be presented are operational even if condition
(\ref{eq:ErasureBasicAssumption}) is not satisfied, they are not
admissible because they violate item \ref{enu:The-service-time} of
Algorithm Design Requirements presented in Section \ref{sec:System-Model}.

In \cite{Krikidis-2009-ID233}, Algorithm \ref{alg:Algorithm-Presented-in} was presented.
\begin{algorithm}[t]
\caption{\label{alg:Algorithm-Presented-in}Algorithm Presented in \cite{Krikidis-2009-ID233} }

\begin{enumerate}
\item If $Q_{1}$ is nonempty, node 1 (re)transmits the packet at the head
of $Q_{1}$ until it is received by either node 2 or node 3.
\begin{enumerate}
\item If the packet is received by node 2 and erased at node 3, it is stored
in a queue $H_{2}$ at node 2.
\end{enumerate}
\item If $Q_{1}$ is empty and $H_{2}$ nonempty, node 2 (re)transmits the
packet at the head of queue $H_{2}$ until it is received by node
3.
\item If $Q_{1}$ and $H_{2}$ are empty, node 2 (re)transmits the packet
at the head of queue $Q_{2}$ until it is received by node 4.
\end{enumerate}
\end{algorithm}
This algorithm is not admissible since it violates items \ref{enu:OrderOfTrans} and \ref{enu:The-service-time}
of Algorithm Design Requirements presented in Section \ref{sec:System-Model}.
However, a slight modification presented in Algorithm \ref{alg:BA}
makes this algorithm admissible.

\begin{algorithm}[t]
\caption{\label{alg:BA}Simple Forwarding Algorithm}

In this algorithm, node 2 maintains a single-packet buffer $B_{2}$.
The algorithm
then operates as follows
\begin{enumerate}
\item If $Q_{1}$ is nonempty and $B_{2}$ is empty, node 1 (re)transmits
the packet at the head of $Q_{1}$ until it is received by either
node 2 or node 3.
\begin{enumerate}
\item If the packet is received by node 2 and erased at node 3, it is stored
in buffer $B_{2}$ at node 2.
\end{enumerate}
\item If $B_{2}$ is nonempty, node 2 (re)transmits the single packet in
$B_{2}$ until it is received by node 3.
\item If $Q_{1}$ and $B_{2}$ are empty, node 2 (re)transmits the packet
at the head of queue $Q_{2}$ until it is received by node 4.
\end{enumerate}
\end{algorithm}

The main difference of Algorithm \ref{alg:BA} from Algorithm \ref{alg:Algorithm-Presented-in}
is that if a session (1,3) packet is received by node 2 and erased
at node 3, then node 2 starts re-transmitting immediately the packet
instead of storing it in a buffer and transmitting it when $Q_{2}$
becomes empty. This modification makes the algorithm admissible. Indeed,
items \ref{enu:No-coding}, \ref{enu:OrderOfTrans} of Algorithm Design
Requirements are obviously satisfied. Item \ref{enu:The-service-time}
is also satisfied, as stated in the next proposition.
\begin{proposition}
\label{prop:StochDominance}Algorithm \ref{alg:BA} satisfies item
\ref{enu:The-service-time} of Algorithm Design Requirements, i.e,
if $S_{k}^{nc}$ ($S_{k}^{c}$) are the service times of the $k$th
session (1,3) packet when no cooperation (cooperation) takes place,
then $S_{k}^{c}$ is stochastically smaller than $S_{k}^{nc}.$
\end{proposition}
\begin{proof}
The proof is given in Appendix \ref{sec:Proof-of-PropStochDom}
\end{proof}
Since the only difference between Algorithms \ref{alg:Algorithm-Presented-in}
and \ref{alg:BA} is the order in which packets are transmitted, the
maximum stable arrival rate of primary session (1,3) and the maximum
throughput of secondary session (2,4) packets are the same under both
algorithms and as shown in \cite{Krikidis-2009-ID233} they are given
by the formulas below.
\begin{equation}
0\leq r_{1}<\frac{\left(1-\epsilon_{3}^{2}\right)\left(1-\epsilon_{23}^{1}\right)}{1-\epsilon_{3}^{2}+\epsilon_{3}^{1}-\epsilon_{23}^{1}},\label{eq:stabRate3}
\end{equation}

\begin{equation}
r_{2}\leq\left(1-r_{1}\frac{1-\epsilon_{3}^{2}+\epsilon_{3}^{1}-\epsilon_{23}^{1}}{\left(1-\epsilon_{3}^{2}\right)\left(1-\epsilon_{23}^{1}\right)}\right)\left(1-\epsilon_{4}^{2}\right).\label{eq:Throughput3}
\end{equation}
Hence the throughput region of Algorithm \ref{alg:BA} is
\[
{\mathcal R}_{\ref{alg:BA}}=\left\{ \left(r_{1},r_{2}\right)\geq\boldsymbol{0}:\ \frac{1-\epsilon_{3}^{2}+\epsilon_{3}^{1}-\epsilon_{23}^{1}}{\left(1-\epsilon_{3}^{2}\right)\left(1-\epsilon_{23}^{1}\right)}r_{1}+\frac{r_{2}}{1-\epsilon_{4}^{2}}\leq1\right\}.
\]

In Figure \ref{fig:The-region-under1,3} we plot the regions ${\mathcal R}_{\ref{alg:NC-Alg}},\ {\mathcal R}_{\ref{alg:BA}}$
when $\epsilon_{3}^{1}=.8$ and $\epsilon_{2}^{1}=\epsilon_{3}^{2}=\epsilon_{4}^{2}=.2$
and all erasure events are independent (hence $\epsilon_{23}^{1}=\epsilon_{2}^{1}\times\epsilon_{3}^{1}=.16$).
We see that by employing Algorithm \ref{alg:BA} the stable arrival
rate for the primary channel increases from .2 to .45. Moreover, for
any value of $r_{1}$ smaller than .45 the maximum throughput for
the secondary channel increases significantly. In the next section
we develop algorithms that increase further the throughput region
of the system by employing network coding techniques.

\begin{figure}
\includegraphics[scale=0.46]{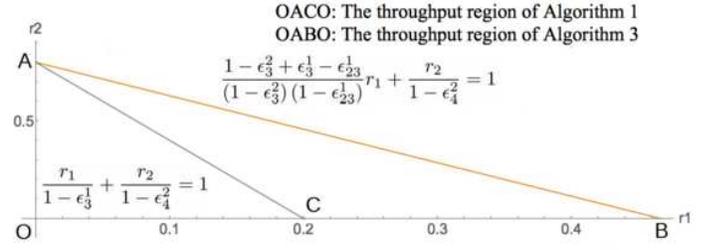}

\caption{\label{fig:The-region-under1,3}The throughput regions of Algorithms
\ref{alg:NC-Alg} and \ref{alg:BA}}
\end{figure}

\section{Proposed Network Coding Algorithm}

\label{subsec:Basic-Network-Coding}

In this section we propose an admissible scheduling algorithm that
at appropriate times, depending on events occurring during the operation,
transmits network-coded packets. The proposed algorithm is admissible
and enhances the maximum throughput of secondary session (2,4), while
leaving the maximum throughput of primary session (1,3) achieved by
Algorithm \ref{alg:BA} unaltered. For the operation of the proposed
algorithm, the following structures are maintained at the nodes.
\begin{enumerate}
\item Two single-packet buffers at node 2, denoted as\footnote{For easy reference we use the following convention in the notation: storage element
$X_{i,k\bar{l}}^{j}$ holds packets located at node $i,$ that have
been transmitted by node $j,$ are received by node $k$ and erased
by node $l$. If the superscript is missing, e.g., $X_{i,k\bar{l}}$
then $X_{i,k\bar{l}}$ holds packets that originated at node $i$.} $B_{2,\bar{3}\bar{4}}^{1}$ and $B_{2,\bar{3}4}^{1},$ for storing
certain packets of session (1,3) transmitted by node 1 and received
by node 2. Buffer $B_{2,\bar{3}\bar{4}}^{1}$ holds packets that are
received by node 2 and erased at 3, 4 . Buffer $B_{2,\bar{3}4}^{1}$
holds packets that are received by nodes 2, 4 and erased at node 3.
The operation of the algorithm ensures that these buffers hold at
most one packet. Moreover, at most one of these buffers may be nonempty
at the beginning of each time slot.
\item One infinite-size queue at node 2, denoted as $Q_{2,3\bar{4}}$, for
storing packets of session (2,4) transmitted by node 2, received by
node 3 and erased at node 4.
\item One single-packet buffer and node 4, denoted as $B_{4,\bar{3}}^{1}$,
for storing packets of session (1,3) transmitted by node 1, erased
at node 3 and received by node 4. The operation of the algorithm ensures
that if buffer $B_{2,\bar{3}4}^{1}$ contains one packet, this packet
is also stored in $B_{4,\bar{3}}^{1}$ at node 4.
\item One infinite-size queue at node 3, denoted as $Q_{3,\bar{4}}^{2}$,
for storing packets of session (2,4) transmitted by node 2, erased
at node 4 and received by node 3. The operation of the algorithm ensures
that the contents of $Q_{3,\bar{4}}^{2}$ are the same as those of
$Q_{2,3\bar{4}}$.
\end{enumerate}
Next we present the details of the operation of the algorithm. Depending
on the status of a transmitted packet at each of the nodes (reception
or erasure) various actions are taken by the nodes. Since each node
sends (ACK, NACK) feedback that is heard by the rest of them, the
nodes are able to perform the actions required by the algorithm. In addition the state of $Q_{1}$ (empty or nonempty) can be
obtained by node 2 by sensing the channel.
\begin{itemize}
\item If $Q_{1},$ $B_{2,\bar{3}\bar{4}}^{1}$ , $B_{2,\bar{3}4}^{1}$ are
all empty, implying that there are no session (1,3) packets in the
network, node 2 (re)transmits the packet at the head of $Q_{2}$ until
it is received by at least one of the nodes 3, 4. If the packet is
received by node 4, it is removed from $Q_{2}.$ If the packet is
received node 3 and erased at node 4, it is removed from $Q_{2}$
and placed in $Q_{2,3\bar{4}}$; also, node 3 stores the packet in
$Q_{3,\bar{4}}^{2}$. As will be explained shortly, the packets stored
in $Q_{2,3\bar{4}}$ are candidates for network coding and are used
by node 2 to form network-coded packets during the times that $Q_{1}$
is nonempty.
\item If queue $Q_{1}$ is nonempty and buffers $B_{2,\bar{3}\bar{4}}^{1}$
and $B_{2,\bar{3}4}^{1}$ are empty, which implies that no packet
of session (1,3) is stored at node 2, node 1 (re)transmits the packet
at the head of $Q_{1}$ until it is received by at least one of the
nodes 2, 3, say at time $t$. During this process, if node 4 receives
the transmitted packet, it stores it in buffer $B_{4,\bar{3}}^{1}$.
If at time $t$ the transmitted packet is received by node 3, the
packet at node $B_{4,\bar{3}}^{1}$ (if any) is removed. If at time
$t$ the packet is erased at node 3 and received by node 2, the packet
is placed in $B_{2,\bar{3}\bar{4}}^{1}$ if $B_{4,\bar{3}}^{1}$ is
empty (i.e., node 4 has not received the packet), and in $B_{2,\bar{3}4}^{1}$ otherwise; in this case, the packet
is removed form $Q_{1}$ and node 2 starts the attempt to deliver
the packet (stored in one of the buffers $B_{2,\bar{3}\bar{4}}^{1},\ B_{2,\bar{3}4}^{1}$
) until it is received by node 3 as described next. Observe that at
time $t$ only one of $B_{2,\bar{3}\bar{4}}^{1}$ and $B_{2,\bar{3}4}^{1}$
can be nonempty. Moreover, if $B_{2,\bar{3}4}^{1}$ is nonempty,
$B_{4,\bar{3}}^{1}$ contains the same packet.
\begin{itemize}
\item If $B_{2,\bar{3}\bar{4}}^{1}$ is nonempty (hence $B_{2,\bar{3}4}^{1}$
is empty), node 2 transmits the packet in $B_{2,\bar{3}\bar{4}}^{1}$
until it is received by at least one of the nodes 3, 4. At this time,
the packet is removed from $B_{2,\bar{3}\bar{4}}^{1}$. Moreover,
if the packet is erased by node 3 and received by node 4, it is moved
to $B_{2,\bar{3}4}^{1}$ and it is also placed in $B_{4,\bar{3}}^{1}.$
We see again that at time $t$ only one of $B_{2,\bar{3}\bar{4}}^{1}$
and $B_{2,\bar{3}4}^{1}$can be nonempty.
\item If $B_{2,\bar{3}4}^{1}$ is nonempty (hence $B_{2,\bar{3}\bar{4}}^{1}$
is empty and $B_{4,\bar{3}}^1$ is nonempty) then,
\begin{itemize}
\item if $Q_{2,3\bar{4}}$ is empty (hence $Q_{4,\bar{3}}^2$ is also empty), node 2 transmits the packet in $B_{2,\bar{3}4}^{1}$
until it is received by node 3, at which time the packet is removed
from $B_{2,\bar{3}4}^{1}$ and $B_{4,\bar{3}}^{1}.$
\item if $Q_{2,3\bar{4}}$ is nonempty (hence $Q_{3,\bar{4}}^{2}$ is nonempty
as well), then the opportunity for network coding arises. Indeed,
let $q_{1}$ and $q_{2}$ be the packets stored in $B_{2,\bar{3}4}^{1}$
and $Q_{2,3\bar{4}}$ respectively. Packet $q_{1}$ is session (1,3)
packet, unknown to node 3 and received by node 4 (it is the packet
stored in $B_{4,\bar{3}}^{1}$). Packet $q_{2}$ is a session (2,4)
packet unknown to node 4 and received by node 3 (it is the packet
stored in $Q_{3,\bar{4}}^{2})$. Hence node 2 sends packet $q=q_{1}\oplus q_{2}$,
where $\oplus$ denotes XOR operation, and if any node in $\left\{ 2,4\right\} $
receives $q,$ that node can decode the packet destined to it. For
example, if node 3 receives packet $p,$ then $q_{1}=p\oplus q_{2}.$
\end{itemize}
\end{itemize}
\end{itemize}
The detailed description of the algorithm, Algorithm 4, is given in Appendix \ref{sec:BNC}. Algorithm
4 is admissible. In fact, the order and service times
of session (1,3) packets are exactly the same as in Algorithm \ref{alg:BA}.
The only difference is that at certain times during which $Q_{1}$
is nonempty, specifically at step \ref{enu:net-code} of the algorithm,
some of these packets are network-coded with packets of session (2,4).
This network coding operation does not alter the time the packet is
delivered to node 3, but allows the increase of throughput for packets
of session (2,4) by allowing for the possibility of simultaneous reception
of packets by nodes 3, 4, using a single transmission by node 2.

\subsection{Performance Analysis of Network Coding Algorithm}

In this section we calculate the throughput region of Algorithm 4.
We first provide an outline of the analysis. For a session
(1,3) packet $q$, let $t_{s}^{q}$ and $t_{r}^{q}$ be respectively
the time when node 1 starts transmitting the packet and the time node
3 receives it - note that according to the algorithm the packet may
have been transmitted to node 3 by node 2. The ``service time''
of the packet is then $t_{r}^{q}-t_{s}^{q}$. Due to the operation
of Algorithm 4 and the statistical assumptions, all session
(1,3) packets have the same distribution of service time. We denote
by $\bar{S}_{1}^{(4)}$ the expected value of the service
time of a session (1,3) packet and provide a method for calculating
it. The queue $Q_{1}^{S}$, discussed in Section \ref{subsec:Prelim},
consisting of all session (1,3) packets that are in the system (at
node 1 and/or node 2), may be viewed as a discrete time queue with
average packet service time $\bar{S}_{1}^{(4)}$. Based
on this observation the maximum packet arrival rate $\lambda_{1}$
for which queue $Q_{1}^{S}$ is stable is given by,
\[
0\leq\lambda_{1}=r_{1}<\frac{1}{\bar{S}_{1}^{4}}\triangleq\mu_{1}^{(4)}.
\]
Next, given $\lambda_{1}$, we calculate the throughput for session
(2,4) packets. For this, we observe that queue $Q_{2,3\bar{4}}$ is
of the ``generic type'' discussed at the end of Section \ref{subsec:Prelim},
where $T_{k}$ is the time when the $k$th busy period of queue $Q_{1}^{S}$
starts. Hence the throughput of packets entering this queue and delivered
to node 4 can be determined through (\ref{eq:stab2})-(\ref{eq:stab3})
after calculating the parameters involved in these formula. The throughput
of session (2,4) packets is then the sum of the throughput of packets
entering $Q_{2,3\bar{4}}$ and the throughput of packets delivered
by node 2 directly to node 4 during the times when queue $Q_{1}^{S}$
is empty.

We now proceed with the detailed analysis. Since as mentioned
in Section \ref{subsec:Basic-Network-Coding} the service times of
packets under Algorithm 4 are the same as those induced
by Algorithm \ref{alg:BA}, we immediately conclude from (\ref{eq:stabRate3})
that
\begin{equation}
\mu_{1}^{\left(4\right)}=\mu_{1}^{\left(\ref{alg:BA}\right)}=\frac{\left(1-\epsilon_{3}^{2}\right)\left(1-\epsilon_{23}^{1}\right)}{1-\epsilon_{3}^{2}+\epsilon_{3}^{1}-\epsilon_{23}^{1}}.\label{eq:rate1NetCoding}
\end{equation}
For the purposes of calculating the appropriate parameters of $Q_{2,3\bar{4}}$
needed in formulas (\ref{eq:stab2})-(\ref{eq:stab3}) we
need to examine the service times of session (1,3) packets under Algorithm
4 in more detail. From its operation it can be seen that
the service of a packet under Algorithm 4 has the same
distribution as the length of time needed for successive returns to
state $1$ of the Markov Chain described in Figure \ref{Fig:Markov_chain3}.
To see this assume that node 1 begins transmission of a new packet
from $Q_{1},$ hence the Markov Chain is in state $1$. At this state:
\begin{itemize}
\item If the packet (sent from $Q_1$) is erased at node 3, and received by nodes 2, 4, an event
with probability $\epsilon_{3}^{1}-\epsilon_{23}^{1}-\epsilon_{34}^{1}+\epsilon_{234}^{1}$, then the packet is stored in buffers $B_{2,\bar{3}4}^{1}$ and $B_{4,\bar{3}}^1 $ and node
2 begins transmission of the packet in $B_{2,\bar{3}4}^{1}$ in the next slot (note that if queue $Q_{2,3\bar{4}}$ is nonempty, the packet in $B_{2,\bar{3}4}^{1}$ is transmitted network-coded with the head of line of packet of $Q_{2,3\bar{4}}$) , i.e., the chain moves to state $3$. At this state:
\begin{itemize}
\item If upon transmission by node 2 the packet is received by node 3, an
event of probability $1-\epsilon_{3}^{2}$, the service time of the
packet completes and we return to state $1.$
\end{itemize}
\item If the packet (sent from $Q_1$) is erased at nodes 3, 4 and received by node 2, an event
with probability $\epsilon_{34}^{1}-\epsilon_{234}^{1}$, then the
packet is stored in buffer $B_{2,\bar{3}\bar{4}}^{1}$ and node 2
begins transmission of the packet in $B_{2,\bar{3}\bar{4}}^{1}$ in the next slot, i.e., the chain moves to state
$2$.
\end{itemize}
Proceeding in a similar fashion we evaluate all the transition probabilities
of the Markov Chain.

\begin{figure}[t]
\label{tab:1}
\centering\includegraphics[scale=0.63]{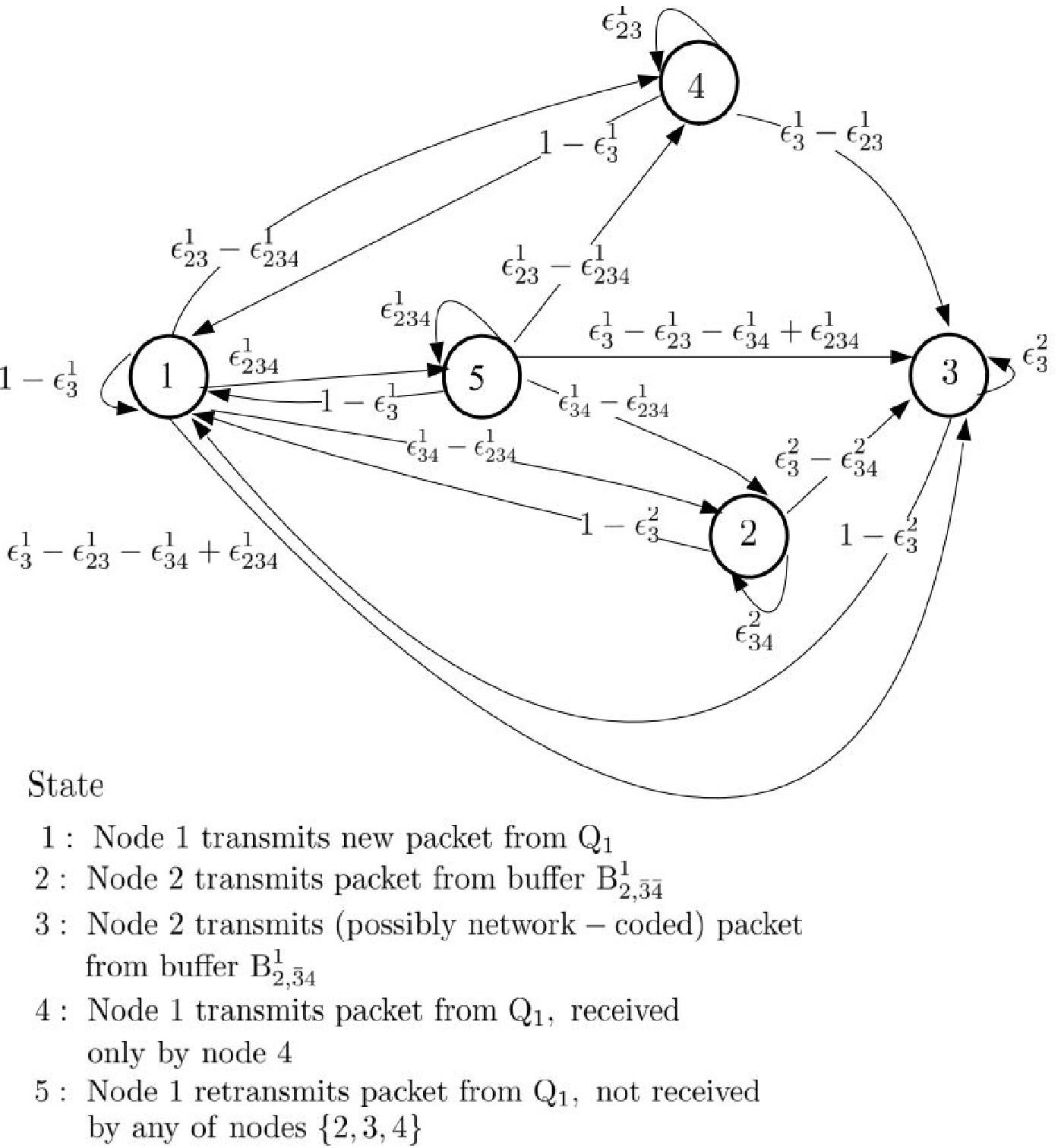}
\protect\caption{Markov chain describing service times of Algorithm 4 }

\label{Fig:Markov_chain3}
\end{figure}

Let $\pi_{k}$ be the steady-state probability that the Markov Chain represented in Figure
\ref{Fig:Markov_chain3} is in state $k\in\left\{ 1,2,3,4\right\} $.
Let $V_{3}$ the number of visits to state $3$ between two successive
visits to state $1$. It is known \cite[page 161 ]{Wolff89} that
the following equalities hold.
\begin{equation}
\mu_{1}^{(4)}=\pi_{1},\label{eq:RateEqSS}
\end{equation}
\begin{equation}
\label{eq:returns}
\mathbb{E}\left[V_{3}\right]=\frac{\pi_{3}}{\pi_{1}}.
\end{equation}

From (\ref{eq:BusyPeriod}), the average length of the busy and idle
period of queue $Q_{1}^{S}$ are given respectively by
\begin{eqnarray*}
B_{1} & = & \frac{\lambda_{1}/\pi_{1}}{\left(1-\lambda_{1}/\pi_{1}\right)(1-q_{0})},\\
I_{1} & = & \frac{1}{(1-q_{0})}.
\end{eqnarray*}

We now concentrate on queue $Q_{2,3\bar{4}}$. This queue is of the generic type discussed at the end of Section \ref{subsec:Prelim}. Specifically, we identify $T_k$ with the beginning of the busy period of queue $Q_{1}^{S}$. Packets arrive to
queue  $Q_{2,3\bar{4}}$ during the idle periods of $Q_{1}^{S}$ when node 2 transmits
a session (2,4) packet that is erased at node 4 and received at node
3. We identify the number of these packets with the parameter $A_{0}$
of the generic queue discussed in Section \ref{subsec:Prelim}, hence,
\begin{equation}
\mathbb{E}\left[A_{0}\right]=I_{1}\left(\epsilon_{4}^{2}-\epsilon_{34}^{2}\right).\label{eq:ArQ2}
\end{equation}
Opportunities to transmit packets from $Q_{2,3\bar{4}}$ arise whenever
buffer $B_{2,\bar{3}4}^{1}$ is nonempty, i.e., the Markov chain in
Figure \ref{Fig:Markov_chain3} is in state 3. Let $V_{3,k}$ be the
number of times state 3 is visited during the service time, $S_{k}$,
of the $k$th packet in a busy period of $Q_{1}^{S}$. The random variables $V_{3,k},\ k=1,2,..$ are i.i.d. and from the definition of the Markov Chain in Figure \ref{Fig:Markov_chain3} it follows that their mean is
\begin{equation}
\label{eq:avreturn}
  \mathbb{E}\left[V_{3,k}\right]=\mathbb{E}\left[V_{k}\right]=\frac{\pi_{3}}{\pi_{1}}.
\end{equation}
Let $N_{B}$
be the number of session (1,3) packets served during a busy period of $Q_{1}^{S}$.
It is known \cite{Wolff89} that
\begin{equation}
\label{eq:number}
B_{1}=\bar{S}_{1}^{4}\mathbb{E}\left[N_{B}\right]=\frac{\mathbb{E}\left[N_{B}\right]}{\pi_{1}}.
\end{equation}

The number of slots available for transmission of session (2,4) packets
during a busy period is
\[
H_{0}=\sum_{k=1}^{N_{B}}V_{3,k}.
\]
Using the fact that $N_{B}$ is a stopping time we obtain from Wald's
equality \cite{Wolff89}, (\ref{eq:avreturn}) and (\ref{eq:number}),
\begin{eqnarray}
\mathbb{E}\left[H_{0}\right] & = & \mathbb{E}\left[V_{3,1}\right]\mathbb{E}\left[N_{B}\right]\nonumber \\
 & = & \pi_{3}B_{1}.\label{eq:AvQ2}
\end{eqnarray}
The service time of a session (2,4) packet transmitted whenever buffer
$B_{2,\bar{3}4}^{1}$ is nonempty, is geometrically distributed with
parameter $1-\epsilon_{4}^{2}$, hence
\begin{equation}
\mathbb{E}\left[S_{1}\right]=\frac{1}{1-\epsilon_{4}^{2}}.\label{eq:SerQ2}
\end{equation}
Also, since $T_{k+1}-T_{k}$ is the sum of the lengths of the $k$th
busy and $k$th idle period of queue $Q_1^S$, we have
\begin{eqnarray}
\mathbb{E}\left[G_{0}\right] & = & B_{1}+I_{1}\nonumber \\
 & = & \frac{\lambda_{1}/\pi_{1}}{\left(1-\lambda_{1}/\pi_{1}\right)(1-q_{0})}+\frac{1}{(1-q_{0})}\nonumber \\
 & = & \frac{1}{\left(1-\lambda_{1}/\pi_{1}\right)(1-q_{0})}.\label{eq:LenQ2}
\end{eqnarray}

Using (\ref{eq:ArQ2}), (\ref{eq:AvQ2}), (\ref{eq:SerQ2}) and (\ref{eq:LenQ2})
above in formulas (\ref{eq:stab2})-(\ref{eq:stab3}) for the generic
queue we obtain after some algebra the following formula for the throughput
of packets in queue $Q_{2,3\bar{4}},$ $r_{q}$.

\begin{equation}
{\rm If}\ \epsilon_{4}^{2}-\epsilon_{34}^{2}\leq\frac{\frac{\pi_{3}\lambda_{1}}{\pi_{1}}(1-\epsilon_{4}^{2})}{1-\frac{\lambda_{1}}{\pi_{1}}},\ r_{q}=\left(\epsilon_{4}^{2}-\epsilon_{34}^{2}\right)\left(1-\frac{\lambda_{1}}{\pi_{1}}\right),\label{eq:stab2-1}
\end{equation}
\begin{equation}
{\rm If\ }\epsilon_{4}^{2}-\epsilon_{34}^{2}>\frac{\frac{\pi_{3}\lambda_{1}}{\pi_{1}}(1-\epsilon_{4}^{2})}{1-\frac{\lambda_{1}}{\pi_{1}}},{\rm }r_{q}=\frac{\pi_{3}}{\pi_{1}}\left(1-\epsilon_{4}^{2}\right)\lambda_{1},\label{eq:stab3-1}
\end{equation}
where the steady state steady state probabilities $\pi_{1},\ \pi_{3}$
can be calculated using the transition probabilities of the Markov
Chain in Figure \ref{Fig:Markov_chain3}. In fact, from (\ref{eq:rate1NetCoding}),
(\ref{eq:RateEqSS}) immediately have,
\begin{equation}
\pi_{1}=\frac{\left(1-\epsilon_{3}^{2}\right)\left(1-\epsilon_{23}^{1}\right)}{1-\epsilon_{3}^{2}+\epsilon_{3}^{1}-\epsilon_{23}^{1}},\label{eq:prob_state_3-0}
\end{equation}
while calculation using the transition probabilities shows that
\begin{equation}
\pi_{3}=1-\frac{\left(\frac{1-\epsilon_{234}^{1}}{1-\epsilon_{23}^{1}}+\frac{\epsilon_{34}^{1}-\epsilon_{234}^{1}}{1-\epsilon_{34}^{2}}\right)}{\left(1-\epsilon_{234}^{1}\right)}\frac{\left(1-\epsilon_{23}^{1}\right)\left(1-\epsilon_{3}^{2}\right)}{\left(1+\epsilon_{3}^{1}-\epsilon_{3}^{2}-\epsilon_{23}^{1}\right)}.\label{eq:prob_state_3-1}
\end{equation}

The throughput of session (2,4) packets transmitted during an idle
period of queue $Q_{1}^{S}$, $r_{d},$ is easily calculated as
\begin{eqnarray}
r_{d} & = & \frac{I_{1}\left(1-\epsilon_{4}^{2}\right)}{B_{1}+I_{1}}\nonumber \\
 & = & \left(1-\lambda_{1}/\pi_{1}\right)\left(1-\epsilon_{4}^{2}\right).\label{eq:rd}
\end{eqnarray}
Since the throughput of session (2,4) packets is $r_{2}=r_{q}+r_{d},$
we conclude from (\ref{eq:stab2-1}), (\ref{eq:stab3-1}), (\ref{eq:prob_state_3-0}),
(\ref{eq:prob_state_3-1}) and (\ref{eq:rd}) the following proposition.

\begin{proposition} \label{prop:RRBAlg} The throughput region of
Algorithm 4  is the set of throughput pairs $\left(r_{1},r_{2}\right)$
satisfying the following inequalities
\begin{eqnarray*}
\frac{1-\epsilon_{3}^{2}+\epsilon_{3}^{1}-\epsilon_{23}^{1}}{\left(1-\epsilon_{3}^{2}\right)\left(1-\epsilon_{23}^{1}\right)}r_{1}+\frac{r_{2}}{1-\epsilon_{34}^{2}} & \leq & 1,\\
\left(\frac{\epsilon_{34}^{1}-\epsilon_{234}^{1}}{\left(1-\epsilon_{34}^{2}\right)\left(1-\epsilon_{234}^{1}\right)}+\frac{1}{1-\epsilon_{23}^{1}}\right)r_{1}+\frac{r_{2}}{1-\epsilon_{4}^{2}} & \leq & 1,
\end{eqnarray*}
\[
r_{i}\geq0,\ i\in\left\{1,2\right\}.
\]
\end{proposition}
\begin{figure}
\includegraphics[scale=0.40]{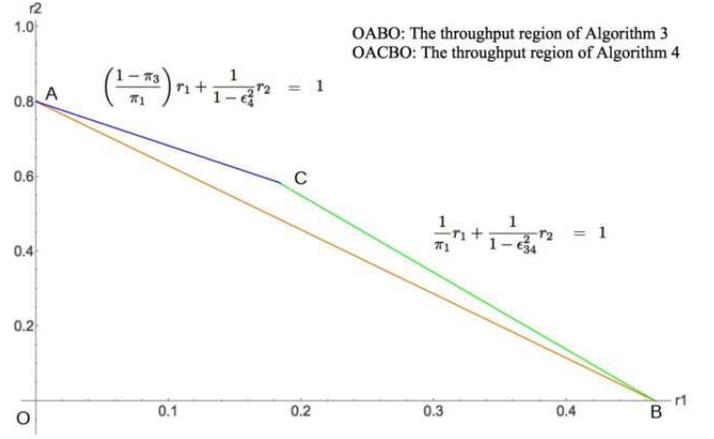}

\caption{\label{fig:The-region-under3,4}The throughput regions of Algorithms
\ref{alg:BA} and 4}
\end{figure}

In Figure \ref{fig:The-region-under3,4} we show the throughput region of Algorithms
\ref{alg:BA} and 4 using the same erasure probability parameters as in Figure \ref{fig:The-region-under1,3}. We see that when node 3 performs network coding, for the same arrival rate of session (1,3) packets, the throughput of Secondary session (2, 4) is increasing, adding in affect the area ABC to the throughput region of the system. We note that this is achieved without adding any additional complexity to the Primary transmitter. The Primary receiver has the additional complexity of storing received session (2,4) packets and performing simple decoding of network-coded packets; this seems an acceptable trade-off for the primary session (1,3), since as is seen in Figure \ref{fig:The-region-under1,3}, cooperation with the secondary session increases significantly the stability region of session (1,3).

\section{An Algorithm with Increased Throughput Region}

\label{sec:Improved}

In this section we examine whether the throughput region of the system
can be increased further by employing more sophisticated network coding
operations. The rationale is the following.

Consider the case where Primary transmitter (node 1) sends a session
(1,3) packet $p_{1}$, and assume that this packet is received only
by Secondary transmitter (node 2). According to Algorithms \ref{alg:BA}
and 4, node 2 will then act as relay for packet $p_{1}$.
Note that for a given $r_{1}$, the increase in $r_{2}$ induced by
Algorithm 4 as compared to $r_{2}$ induced by Algorithm
\ref{alg:BA}, occurs because, during the attempt by node 2 to
send packet $p_{1}$, it happens that this packet has already been
received by node $4$; so the possibility of network coding operation
arises. However, if $\epsilon_{34}^{1}<\epsilon_{34}^{2}$, then it
is more likely that packet $p_{1}$ is received by either node 3 or
node 4 if it is re-transmitted by node 1. On the other hand, if the rate of
packets to node 1, $r_{1}=\lambda_{1}$, is close to point B in Figure
\ref{fig:The-region-under1,3} this re-transmission should be avoided
since queue $Q_{1}$ will become unstable. Therefore, it seems that,
in order to effect increase in the throughput of session (2,4) while
maintaining admissibility of the algorithm, a compromise between the following
two cases must be made: a) node 2 acts immediately as a relay of packet $p_{1}$
and b) node 1 keeps re-transmitting $p_{1}$ until received by either
node 3 or node 4.

To effect this compromise, we modify Algorithm 4 as follows.
We introduce a parameter $q,\ 0\leq q\leq1$. When node 1 transmits
a packet $p_{1}$ that is seen only by node 2 (hence now the packet
is stored in buffer $B_{2,\bar{3}\bar{4}}^{1}$ ) then $p_{1}$ remains
in $Q_{1}$ and is transmitted by node 1 with probability $q$ and
by node 2 with probability $1-q.$ In both cases, if $p_{1}$ is received
by node 3, then it is removed from $Q_{1}$ and $B_{2,\bar{3}\bar{4}}^{1}$;
if on the other hand it is erased at node 3 but received by node 4,
then the packet is removed from $Q_{1}$ and node 2 acts as relay
for $p_{1}$ as in Algorithm 4.

The detailed description of the algorithm, Algorithm 5, is given in Appendix \ref{sec:ImprovedApp}.
This algorithm differs from Algorithm 4 in four places
(these places are written in italics in the algorithm).
\begin{enumerate}
\item In Item \ref{it:case1}, if the packet transmitted by node 1 is received
only by node 2 and node 4 has not seen the packet earlier, then the
packet is placed in buffer $B_{2,\bar{3}\bar{4}}^{1}$ but it also
remains at the head of line of $Q_{1}$.
\item In Item \ref{it:case2}, instead of node 2 transmitting the packet
in $B_{2,\bar{3}\bar{4}}^{1}$ as in Algorithm 4, the
packet is transmitted either by node 1 or by node 2 with probabilities
$q$ and $1-q$ respectively.
\item In Item \ref{it:case3}, the packet that is received by node 3 is
removed from buffer $B_{\bar{3}\bar{4}}^{1}$ and also from queue
$Q_{1}$ (since it has not been removed in Item \item In Item \ref{it:case4}, the packet is removed from both buffer $B_{2,\bar{3}\bar{4}}^{1}$
and queue $Q_{1}$ (since it has not been removed in Item \ref{it:case1}),
and from now on node 2 will act as relay for the packet.
\end{enumerate}

The service times of packets under algorithm Algorithm 5
may increase as compared to the service times of the packets under
Algorithm 4, but it can be shown by employing arguments
similar to those used in the proof of Proposition \ref{prop:StochDominance}
that they remain stochastically smaller than the service times of
Algorithm \ref{alg:NC-Alg} that involved no cooperation.

\begin{figure}
\centering\includegraphics[scale=0.62]{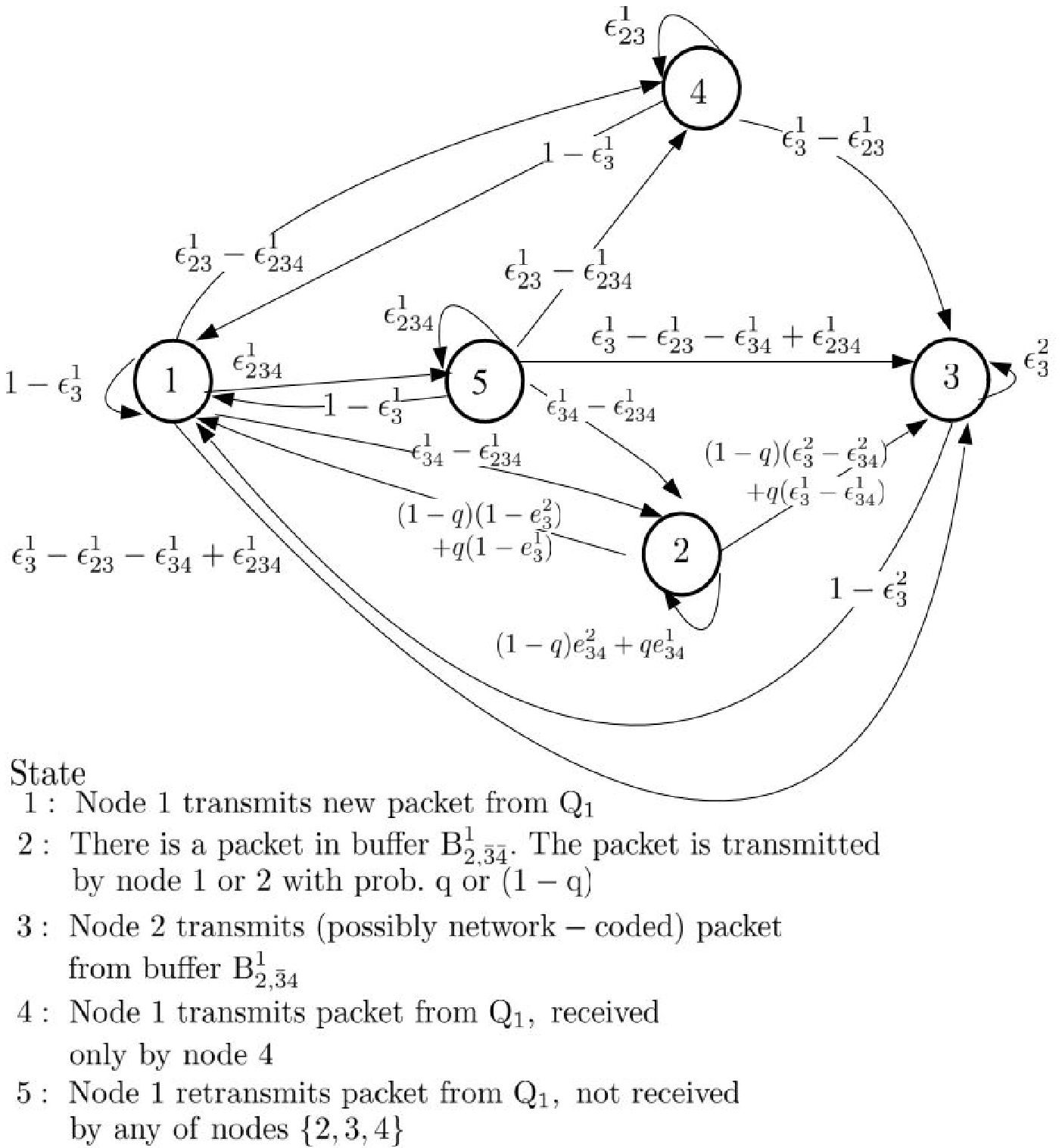}
\caption{The Markov chain describing service times of Algorithm 5 }
\label{fig:ImprovedMC}
\end{figure}

For given $q$, the performance analysis of Algorithm 5
is similar to the performance Analysis of Algorithm 4.
The main difference is that the Markov Chain describing the service
times of packets under Algorithm 5 is now the one
presented in Figure \ref{fig:ImprovedMC}.
Let $\pi_{i}(q),1\leq i\leq5,$  be the steady state probabilities of this
Markov chain when parameter $q$ is used. We calculate below two formulas that are needed to determine the throughput region of the system.

\begin{equation}
\frac{1}{\pi_{1}(q)}=\frac{1+\epsilon_{3}^{1}-\epsilon_{3}^{2}-\epsilon_{23}^{1}}{\left(1-\epsilon_{3}^{2}\right)\left(1-\epsilon_{23}^{1}\right)}+C_{1}(q)\label{eq:1overp1}
\end{equation}
\begin{equation}
\frac{1-\pi_{3}(q)}{\pi_{1}(q)}=\frac{\epsilon_{34}^{1}-\epsilon_{234}^{1}}{\left(1-\epsilon_{34}^{2}\right)\left(1-\epsilon_{234}^{1}\right)}+\frac{1}{1-\epsilon_{23}^{1}}-C_{2}(q),\label{eq:1-p3}
\end{equation}
where
\[
C_{1}(q)=\frac{\left(\epsilon_{3}^{1}-\epsilon_{3}^{2}\right)\left(\epsilon_{34}^{1}-\epsilon_{234}^{1}\right)}{\left(1-\epsilon_{234}^{1}\right)\left(1-\epsilon_{3}^{2}\right)}\frac{q}{1-q\epsilon_{34}^{1}-\left(1-q\right)\epsilon_{34}^{2}},
\]
\[
C_{2}(q)=\frac{\left(\epsilon_{34}^{1}-\epsilon_{234}^{1}\right)\left(\epsilon_{34}^{2}-\epsilon_{34}^{1}\right)}{\left(1-\epsilon_{234}^{1}\right)\left(1-\epsilon_{34}^{2}\right)}\frac{q}{1-q\epsilon_{34}^{1}-\left(1-q\right)\epsilon_{34}^{2}}.
\]

We can now state the following
proposition which is analogous to Proposition \ref{prop:RRBAlg}.


\begin{proposition} \label{prop:ImprovedProp} The
throughput region of Algorithm 5 is the set
of pairs $\left(r_{1},r_{2}\right)$ satisfying the following
inequalities

\begin{align*}
\left(\frac{1+\epsilon_{3}^{1}-\epsilon_{3}^{2}-\epsilon_{23}^{1}}{\left(1-\epsilon_{3}^{2}\right)\left(1-\epsilon_{23}^{1}\right)}+C_{1}(q)\right)r_{1}+\frac{r_{2}}{1-\epsilon_{34}^{2}} & \leq1,\\
\left(\frac{\epsilon_{34}^{1}-\epsilon_{234}^{1}}{\left(1-\epsilon_{34}^{2}\right)\left(1-\epsilon_{234}^{1}\right)}+\frac{1}{1-\epsilon_{23}^{1}}-C_{2}(q)\right)r_{1}+\frac{r_{2}}{1-\epsilon_{4}^{2}} & \leq1,
\end{align*}
\[
0\leq q\leq1,\ r_{i}\geq0,\ i=1,2.
\]
\end{proposition}

The throughput region of Algorithm 5 can be 
as follows. For given arrival rate
$r_{1}=\lambda_{1}$ where $0\leq r_{1}\leq\pi_{1}$,
we calculate the parameter $q$ that maximizes the throughput $r_{2}\geq0$
under the constraints described in Proposition \ref{prop:ImprovedProp},
i.e.,
\begin{equation}
r_{2}=\max_{0\leq q\leq1}f(q),\label{eq:fin}
\end{equation}
where
\[
f(q)=\min\left\{ (1-\epsilon_{34}^{2})\left(1-\frac{r_{1}}{\pi_{1}(q)}\right)\right.,
\]\[
\left.(1-\epsilon_{4}^{2})\left(1-\frac{(1-\pi_{3}(q))}{\pi_{1}(q)}r_{1}\right)\right\}
\]
This determines a pair $(r_{1},r_{2})$ on the 
of Algorithm 5.

Algorithm 4 is a special case of Algorithm 5
obtained by setting $q=0$. Hence the throughput region of Algorithm
5 is at least as large as the throughput region of
Algorithm 4. The following example shows that in fact
the throughput region of the system may increase under Algorithm 5
in certain cases. Consider the following system parameters $\epsilon_{2}^{1}=.3$
,$\epsilon_{3}^{1}=.77$, $\epsilon_{4}^{1}=.6$, $\epsilon_{3}^{2}=.75$
,$\epsilon_{4}^{2}=.85$, $\epsilon_{23}^{1}=\epsilon_{2}^{1}\times\epsilon_{3}^{1}=.231$,
$\epsilon_{34}^{1}=\epsilon_{3}^{1}\times\epsilon_{4}^{1}=.462$,
$\epsilon_{34}^{2}=.75$, and $\epsilon_{234}^{1}=\epsilon_{2}^{1}\times\epsilon_{3}^{1}\times\epsilon_{4}^{1}=.1386$.
In Figure \ref{fig:ImprovedRegion} we show the throughput regions
of Algorithm 4 and 5 under these parameters.

While as shown in the example above Algorithm 5 has the
potential of increasing the Throughput region of the system, it requires
the careful selection of parameter $q$, which depends both on the
erasure probabilities of the channel and the arrival rate of packets
to the Primary transmitter. However, if it holds $\epsilon_{34}^{1}\geq\epsilon_{34}^{2}$, additionally to the condition
$\epsilon_{3}^{1}\geq\epsilon_{3}^{2}$,
then Algorithm 4 suffices as the next proposition shows.
\begin{proposition} If in addition to the condition $\epsilon_{3}^{1}\geq\epsilon_{3}^{2}$
it holds $\epsilon_{34}^{1}\geq\epsilon_{34}^{2}$, then the throughput
regions of Algorithms 4 and 5 coincide.
\end{proposition} \begin{proof} Notice first that if $\epsilon_{3}^{1}\geq\epsilon_{3}^{2}$,
then the packet service times induced by Algorithm 5
using $q=0$ - i.e., by Algorithm 4 - are stochastically
smaller that the service times of packets induced by Algorithm 5
using any $q\geq0$. This is due to the fact that under Algorithm
4, node 1 re-transmits a packet fewer times than Algorithm
5 using any $q>0$, and can be shown in a manner
similar to the proof of Proposition \ref{prop:StochDominance}. Since
the average service time of a packet is the inverse of successive
returns to state 1 of the Markov chain in Figure \ref{fig:ImprovedMC},
we conclude that
\begin{equation}
\frac{1}{\pi_{1}(0)}\leq\frac{1}{\pi_{1}(q)},\ 0\leq q\leq1.\label{eq:item1}
\end{equation}
An alternative way to establish this inequality is to take the derivative
of the function $\phi(q)=1/\pi_{1}(q)$. From (\ref{eq:1overp1}) it can be seen that
\[
\frac{d\phi}{dq}=\frac{\left(\epsilon_{3}^{1}-\epsilon_{3}^{2}\right)\left(\epsilon_{34}^{1}-\epsilon_{234}^{1}\right)\left(1-\epsilon_{34}^{2}\right)}{\left(1-\epsilon_{3}^{2}\right)\left(\epsilon_{34}^{2}+q\left(\epsilon_{34}^{1}-\epsilon_{34}^{2}\right)-1\right)^{2}\left(1-\epsilon_{234}^{1}\right)},
\]
which is non-negative since $\epsilon_{3}^{1}\geq\epsilon_{3}^{2}$.
Hence the function $\phi(q)$ is non-decreasing and (\ref{eq:item1})
follows. Next, from (\ref{eq:1-p3})
we compute
\begin{equation*}
\frac{1-\pi_{3}(q)}{\pi_{1}(q)}-\frac{1-\pi_{3}(0)}{\pi_{1}(0)}=
\end{equation*}
\begin{equation*}
\frac{q\left(\epsilon_{34}^{1}-\epsilon_{234}^{1}\right)\left(\epsilon_{34}^{1}-\epsilon_{34}^{2}\right)}{\left(1-\epsilon_{234}^{1}\right)\left(1-\epsilon_{34}^{2}\right)\left(1-q\epsilon_{34}^{1}-(1-q)\epsilon_{34}^{2}\right)}
\end{equation*}
and since $\epsilon_{34}^{1}\geq\epsilon_{34}^{2}$ it easily follows
that
\begin{equation}
\frac{1-\pi_{3}(q)}{\pi_{1}(q)}\geq\frac{1-\pi_{3}(0)}{\pi_{1}(0)}.\label{eq:item2}
\end{equation}
From (\ref{eq:item1}), (\ref{eq:item2}) we conclude that $f(0)\geq f(q),\ 0\leq q\leq1$
and hence, according to (\ref{eq:fin}) the throughput region of the
system is maximized when $q=0$. \end{proof}

\begin{figure}
\centering\includegraphics[scale=0.54]{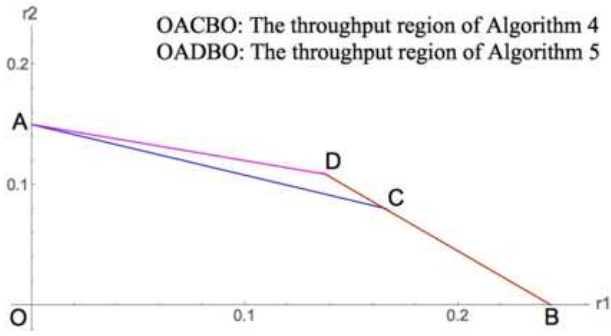}
\caption{The Throughput Regions of Algorithms 4 and 5}
\label{fig:ImprovedRegion}
\end{figure}

\section{Conclusions}
\label{sec:Conclusions}
We proposed two algorithms for Primary-Secondary user cooperation in Cognitive Networks. According to the first algorithm, when a packet sent by the Primary  transmitter is erased at the Primary receiver and received by the Secondary transmitter, the Secondary transmitter acts as relay for the packet. Depending on the channel feedback, the Secondary transmitter may send network-coded packets which allows simultaneous packet reception by the Primary and Secondary receivers. We analyzed the performance of the proposed algorithm. The results show that when compared to the case where the Secondary transmitter acts as a relay without performing network coding, significant improvement of the throughput of the secondary channel may occur.
The algorithm imposes no extra implementation requirements to the Primary transmitter apart from  listening to the feedback sent by the Secondary transmitter. The Primary receiver has the additional requirement that it stores received packets intended for the Secondary receiver and it performs decoding of network-coded packets. In return though, the stability region and the service times of all Primary Channel packets are significantly improved. A notable feature of the algorithm is that no knowledge of packet arrival rates to Primary transmitter and of channel statistics is required as long as it is known that the  Secondary transmitter to Primary receiver channel is better that the channel from the Primary transmitter to Primary receiver.

We next examined the possibility of increasing the throughput region of the system by more sophisticated network coding techniques. We presented a second algorithm which is  a generalization of the first and showed that this increase is possible in certain cases. However, in this case, knowledge of channel erasure probabilities, as well as the arrival rate of Primary transmitter packets are crucial for the algorithm to operate correctly.

It is interesting to examine whether the throughput of the Secondary channel can be increased further by more sophisticated network coding operations. Towards this direction, work is underway to examine the information theoretic capacity of the system and its relation to the throughput region of the current algorithms. Preliminary results show that the these algorithms achieve most of the capacity region of the system.

\appendices

\section{Proof of Proposition \ref{prop:StochDominance}}
\label{sec:Proof-of-PropStochDom}

The proof of stochastic dominance can be done by explicitly calculating the relevant probabilities and then showing the required inequality. We can avoid cumbersome calculations, however, by resorting to a technique commonly used in this type of proofs. Specifically, let $S^{nc}_k,\ S^c_k,\ k=1,2,...$ be the
service times of packet $k$ transmitted by node 1 under  Algorithms \ref{alg:NC-Alg} and \ref{alg:BA} respectively (these
service times are i.i.d. under both algorithms). To show
stochastic dominance, we construct on the same probability space two
random variables $\hat{S}^{nc}$ and $\hat{S}^{c}$ with the following
properties.
\begin{enumerate}
\item It holds $\hat{S}^{c}\leq\hat{S}^{nc}.$
\item $\hat{S}^{c}$and $S_{k}^{c}$ have the same distribution.
\item $\hat{S}^{nc}$and $S_{k}^{nc}$ have the same distribution.
\end{enumerate}
The fact that $S_{k}^{c}$ is stochastically smaller than $S_{k}^{nc}$
follows then immediately from the inequality $\hat{S}^{c}\leq\hat{S}^{nc}.$

We now proceed with the construction of $\hat{S}^{nc}$ and $\hat{S}^{c}$.
Consider on the same probability space a sequence $\left(\hat{Z}_{2}^{1}(t),\hat{Z}_{3}^{1}(t)\right),\ t=0,1,...$
of i.i.d pairs of random variables (for given $t$ the pair $\hat{Z}_{2}^{1}(t),\hat{Z}_{3}^{1}(t)$ may be dependent),
and a sequence $\theta (t),\ t=0,1,...$ of i.i.d. random variables,
independent of $\left(\hat{Z}_{2}^{1}(t),\hat{Z}_{3}^{1}(t)\right),\ t=0,1,....$ All
random variables take values either 0 or 1, with probabilities,
\begin{align*}
\Pr\left(\hat{Z}_{3}^{1}(t)=0\right) & =\epsilon_{3}^{1},\ \Pr\left(\hat{Z}_{2}^{1}(t)=\hat{Z}_{3}^{1}(t)=0\right)=\epsilon_{23}^{1},\\
\Pr\left(\theta (t)=0\right) & =\frac{\epsilon_{3}^{2}}{\epsilon_{3}^{1}},\ t=1,2,...
\end{align*}
Note that $\Pr\left(\theta (t)=0\right)$ is indeed a probability
because of (\ref{eq:ErasureBasicAssumption}). Let also
\begin{equation}
\hat{J} (t)=\left\{ \begin{array}{ccc}
\theta (t) & \mbox{if} & \hat{Z}_{3}^{1}(t)=0\\
\hat{Z}_{3}^{1}(t) & \mbox{if } & \hat{Z}_{3}^{1}(t)=1.
\end{array}\right.\label{eq:Correlate}
\end{equation}
From (\ref{eq:Correlate}) we see that $\hat{J}(t)$ takes values 0 or
1 and $\hat{J}(t)\geq \hat{Z}_{3}^{1}(t),\ t=0,1,2....$ Moreover, $\hat{J} (t),\ t=1.,2,...$ are i.i.d, $\hat{J}(t)$ is independent of $\hat{Z}_{3}^1(\tau),\ \hat{Z}_{3}^2(\tau),\ \tau\neq t$, and
\begin{eqnarray}
\Pr\left(\hat{J}(t)=0\right) & = & \Pr\left(\theta (t)=0,\ \hat{Z}_{3}^{1}(t)=0\right) \nonumber\\
 & = & \Pr\left(\theta (t) =0\right)\epsilon_{3}^{1}\mbox{ by indep. of \ensuremath{\hat{Z}_{3}^{1}(t)}, \ensuremath{\theta (t).}} \nonumber\\
 & = & \epsilon_{3}^{2},\label{eq:erase2to3}
\end{eqnarray}
that is, the random variables $\hat{J}(t),\ t=0,1,...$ are identically distributed to the random variables $Z_3^2(t),\ t=0,1,...$ denoting erasure events defined in Section \ref{sec:System-Model}.

Let $\hat{T}_{23}$
be the stopping time denoting the first time at least one
of the random variables $\hat{Z}_{2}^{1}(t),$
$\hat{Z}_{3}^{1}(t)$ takes the value 1, i.e.,
\begin{equation}
\hat{T}_{23}=\min_{t\geq0}\{\hat{Z}_{2}^{1}(\tau)=\hat{Z}_{3}^{1}(\tau)=0,\ \tau=0,..,t-1,
\end{equation}
\[
\ \hat{Z}_{2}^{1}(t)+\hat{Z}_{3}^{1}(t)>0\}.
\]

Let $\hat{S}^{nc}\geq \hat{T}_{23}$ be the first time that the random variable
$\hat{Z}_{3}^{1}(t)$ takes the value 1, and define $\hat{S}^{c}$ as follows.
If at time $\hat{T}_{23}$ it holds $\hat{Z}_{3}^{1}(\hat{T}_{23})=1$ then $\hat{S}^{c}=\hat{T}_{23}$.
Else $\hat{S}^{c}$ is the first time, $\hat{T}$, after $\hat{T}_{23}$ that the random
variable $\hat{J}(t)$ becomes 1. Therefore it holds,
\begin{equation}
\label{eq:fc}
    \hat{S}^{c}=\hat{T}_{23}+
    \boldsymbol{1}_{\{\hat{Z}_{3}^{1}(\hat{T}_{23})=0\}}
    (\hat{T}-\hat{T}_{23}),
\end{equation}
where $\boldsymbol{1}_{\mathcal{A}}$ is the indicator function of event $\mathcal{A}$.
Notice that the interval $\hat{T}-\hat{T}_{23}$ depends only on $\hat{J}(\hat{T}_{23}+t),\ t\geq 1$ and these variables are independent of $ \hat{T}_{23}$, since $\hat{T}_{23}$ is a stopping time for the sequence $\left(\hat{Z}_3^1(t),\hat{Z}_3^2(t)\right)$. Note also that we can write
\begin{equation}
\label{eq:fnc}
    \hat{S}^{nc}=\hat{T}_{23}+
    \boldsymbol{1}_{\{\hat{Z}_{3}^{1}(\hat{T}_{23})=0\}}
    (\hat{S}^{nc}-\hat{T}_{23}),
\end{equation}
where interval $\hat{S}^{nc}-\hat{T}_{23}$ depends only on $\hat{Z}_{3}^{1}(\hat{T}_{23}+t),\ t\geq 1$
From (\ref{eq:fc}), (\ref{eq:fnc}) and since $\hat{J}(t)\geq \hat{Z}_{3}^{1}(t),\ m=1,2....,$
it follows that,
\[
\hat{S}^{c}\leq \hat{S}^{nc}.
\]
We now examine the service times of Algorithms \ref{alg:NC-Alg} and \ref{alg:BA}. For simplicity we omit the packet index $k$ from $S_k^{nc}$ and $S_k^c$.
According to Algorithm \ref{alg:NC-Alg}, $S^{nc}$ is the first time $Z_{3}^{1}(t)$ takes the value one, where $Z_{k}^{i}(t)$ are the random variables expressing erasure events defined in Section \ref{sec:System-Model}. Since the random variables $Z_{3}^{1}(t)$ and $\hat{Z}_{3}^{1}(t),\ t=1,2,...$ are identically distributed, it follows that $S^{nc}$ and $\hat{S}^{nc}$ have the same distribution.

Let $T_{23}$
be the first time at least one
of the random variables ${Z}_{2}^{1}(t),\ {Z}_{3}^{1}(t)$ takes the value 1. From the operation of Algorithm \ref{alg:BA} it follows that if $Z_{3}^{1}(T_{23})=1$ (the transmitted packet is received by node 3 at time $T_{23}$) then $S^c=T_{23}$. Else (the transmitted packet is received by node 2 and erased at node 3) $S^c$ is the first time, $T$, after $T_{23}$ that the random variable $Z_{3}^{2}(t)$ becomes 1. From the definitions it holds,
\begin{equation}
\label{eq:fcr}
    {S}^{c}={T}_{23}+
    \boldsymbol{1}_{\{{Z}_{3}^{1}({T}_{23})=0\}}
    ({T}-{T}_{23}),
\end{equation}
where interval ${S}^{nc}-{T}_{23}$ depends only on the variables ${Z}_{3}^{2}({T}_{23}+t),\ t\geq 1$.

According to (\ref{eq:erase2to3}), and the definitions, the set of random variables $\hat{T}_{23},\ $, $\hat{J}(\hat{T}_{23}+t),\ t\geq 1$  are identically distributed to the set $T_{23}$,  ${Z}_{3}^{2}({T}_{23}+t),\ t\geq 1$. It then follows from (\ref{eq:fc}) and (\ref{eq:fcr}) that
 $S^{nc}$ and $\hat{S}^{nc}$ have the same distribution.

\section{Algorithm 4: Network Coding Algorithm}
\label{sec:BNC}

In this algorithm node 2 maintains, in addition to queue $Q_2$, two single-packet buffers $B_{2,\bar{3}\bar{4}}^{1}$
and $B_{2,\bar{3}4}^{1}$ and one infinite size queue $Q_{2,3\bar{4}}$.
Node 3 maintains one infinite size queue $Q_{3,\bar{4}}^{2}$ that stores the same packets as $Q_{2,3\bar{4}}$. Also node
4 maintains one single-packet buffer $B_{4,\bar{3}}^{1}$ that stores the same packet as $B_{2,\bar{3}4}^{1}$ if the latter buffer is nonempty. The algorithm
operates as follows:
\begin{enumerate}
\item If $Q_{1}$, $B_{2,\bar{3}\bar{4}}^{1}$ and $B_{2,\bar{3}4}^{1}$
are empty, node 2 transmits the packet at the head of $Q_{2}$.

\begin{enumerate}
\item If the packet is received by node 4, it is removed from $Q_{2}.$
\item If the packet is erased at node 4 and received by node 3, it is removed
from $Q_{2}$ and placed in queue $Q_{2,3\bar{4}}.$ The packet is
also placed in $Q_{3,\bar{4}}^{2}$.
\end{enumerate}
\item If $Q_{1}$ is non-empty, and $B_{2,\bar{3}\bar{4}}^{1}$ and $B_{2,\bar{3}4}^{1}$
are empty, node 1 transmits the packet at the head of queue $Q_{1}$.

\begin{enumerate}
\item If the transmitted packet is received by node 3, it is removed from
$Q_{1}$. If $B_{4,\bar{3}}^{1}$ is non-empty (node 4 has received the packet earlier) the packet in $B_{4,\bar{3}}^{1}$
is removed.
\item If the transmitted packet is erased at node 3:

\begin{enumerate}
\item If the packet is received by both nodes 2 and 4, it is stored at buffers
$B_{2,\bar{3}4}^{1}$ and $B_{4,\bar{3}}^{1}$ and it is removed from
queue $Q_{1}$.
\item If the packet is erased at node 2, received by node 4 and $B_{4,\bar{3}}^{1}$
is empty (node 4 has not received the packet earlier), the packet is placed in $B_{4,\bar{3}}^{1}$.
\item If the packet is received by node 2 and erased at node 4:

\begin{enumerate}
\item If $B_{4,\bar{3}}^{1}$ is empty (node 4 has not received the packet earlier) then the packet is stored at buffer
$B_{2,\bar{3}\bar{4}}^{1}$ and is removed from $Q_{1}$.
\item If $B_{4,\bar{3}}^{1}$ is non-empty (the packet has been received earlier by node 4) then the packet is stored at
buffer $B_{2,\bar{3}4}^{1}$ and it is removed from $Q_{1}$.
\end{enumerate}
\end{enumerate}
\end{enumerate}
    \item
    If $B_{2,\bar{3}\bar{4}}^{1}$ is non-empty (hence $B_{2,\bar{3}4}^{1}$ and $B^{1}_{4,\bar{3}}$
are empty), node 2 transmits the packet stored in buffer $B^{1}_{2,\bar{3}\bar{4}}$.

\begin{enumerate}
\item If the packet is received by node 3, it is removed from buffer $B^{1}_{2,\bar{3}\bar{4}}$.
\item If the packet is erased at node 3 and received by node 4, it is removed
from buffer $B^{1}_{2,\bar{3}\bar{4}}$ and placed in buffers $B^{1}_{2,\bar{3}4}$ and $B^{1}_{4,\bar{3}}$.
\end{enumerate}
\item If $B_{2,\bar{3}4}^{1}$ is non empty (hence $B_{2,\bar{3}\bar{4}}^{1}$
is empty).

\begin{enumerate}
\item If $Q_{2,3\bar{4}}$ is empty (no coding opportunity), node 2 transmits the packet stored
in buffer $B_{2,\bar{3}4}^{1}$.

\begin{enumerate}
\item If the packet is received by node 3, it is removed $B_{2,\bar{3}4}^{1}$.
\end{enumerate}
\item \label{enu:net-code}If $Q_{2,3\bar{4}}$ is non-empty (coding opportunity), node 2
transmits the network-coded packet $p=q_{1}\oplus q_{2}$, where $q_{1}$
is the packet stored in buffers $B_{2,\bar{3}4}^{1}$ and $B_{4,\bar{3}}^{1}$, and
$q_{2}$
is the packet at the head of queues $Q_{2,3\bar{4}}$ and $Q_{3,\bar{4}}^{2}$.

\begin{enumerate}
\item If packet $p$ is received by node 4, node 4 decodes  packet $q_{2}=p\oplus q_1$ and $q_2$ is removed from
$Q_{2,3\bar{4}}$ and $Q_{3,\bar{4}}^{2}$.
\item If packet $p$ is received by node 3, node 3 decodes packet $q_{1}=p\oplus q_2$ and $q_2$ is removed from
buffers $B_{2,\bar{3}4}^{1}$ and $B_{4,\bar{3}}^{1}$ .
\end{enumerate}
\end{enumerate}
\end{enumerate}

\section{Algorithm 5: Network Coding Algorithm with Increased Throughput Region}
\label{sec:ImprovedApp}

In this algorithm node 2 maintains, in addition to queue $Q_2$, two single-packet buffers $B_{2,\bar{3}\bar{4}}^{1}$
and $B_{2,\bar{3}4}^{1}$ and one infinite size queue $Q_{2,3\bar{4}}$.
Node 3 maintains one infinite size queue $Q_{3,\bar{4}}^{2}$ that stores the same packets as $Q_{2,3\bar{4}}$. Also node
4 maintains one single-packet buffer $B_{4,\bar{3}}^{1}$ that stores the same packet as $B_{2,\bar{3}4}^{1}$ if the latter buffer is nonempty. The algorithm
operates as follows:
\begin{enumerate}
\item If $Q_{1}$, $B_{2,\bar{3}\bar{4}}^{1}$ and $B_{2,\bar{3}4}^{1}$
are empty, node 2 transmits the packet at the head of $Q_{2}$.

\begin{enumerate}
\item If the packet is received by node 4, it is removed from $Q_{2}.$
\item If the packet is erased at node 4 and received by node 3, it is removed
from $Q_{2}$ and placed in queue $Q_{2,3\bar{4}}.$ The packet is
also placed in $Q_{3,\bar{4}}^{2}$.
\end{enumerate}
\item If $Q_{1}$ is non-empty, and $B_{2,\bar{3}\bar{4}}^{1}$ and $B_{2,\bar{3}4}^{1}$
are empty, node 1 transmits the packet at the head of queue $Q_{1}$.

\begin{enumerate}
\item If the transmitted packet is received by node 3, it is removed from
$Q_{1}$. If $B_{4,\bar{3}}^{1}$ is non-empty (node 4 has received the packet earlier) the packet in $B_{4,\bar{3}}^{1}$
is removed.
\item If the transmitted packet is erased at node 3:

\begin{enumerate}
\item If the packet is received by both nodes 2 and 4, it is stored at buffers
$B_{2,\bar{3}4}^{1}$ and $B_{4,\bar{3}}^{1}$ and it is removed from
queue $Q_{1}$.
\item If the packet is erased at node 2, received by node 4 and $B_{4,\bar{3}}^{1}$
is empty (node 4 has not received the packet earlier), the packet is placed in $B_{4,\bar{3}}^{1}$.
\item If the packet is received by node 2 and erased at node 4:

\begin{enumerate}

\item \label{it:case1}\emph{ If $B_{4,\bar{3}}^{1}$ is empty (node 4 has not received the packet earlier) then the packet is stored at buffer
$B_{2,\bar{3}\bar{4}}^{1}$.}

 \item If $B_{4,\bar{3}}^{1}$ is non-empty (the packet has been received earlier by node 4) then the packet is stored at
buffer $B_{2,\bar{3}4}^{1}$ and it is removed from $Q_{1}$.
\end{enumerate}
\end{enumerate}
\end{enumerate}
    \item
    \label{it:case2}
    \emph{
    If $B_{2,\bar{3}\bar{4}}^{1}$ is non-empty ( $B_{2,\bar{3}4}^{1}$ and $B^{1}_{4,\bar{3}}$
are empty), node 1 transmits the head of line packet in queue $Q_1$ with probability $q$ of node 2 transmits the (same) packet stored in buffer $B^{1}_{2,\bar{3}\bar{4}}$ with probability $1-q$.}

\begin{enumerate}
\item \label{it:case3}\emph{If the packet is received by node 3, it is removed from buffer $B^{1}_{2,\bar{3}\bar{4}}$ and queue $Q_1$.}
\item \label{it:case4} \emph{If the packet is erased at node 3 and received by node 4, it is removed
from buffer $B^{1}_{2,\bar{3}\bar{4}}$ and queue $Q_1$, and placed in buffers $B^{1}_{2,\bar{3}4}$ and $B^{1}_{4,\bar{3}}$.}
\end{enumerate}
\item If $B_{2,\bar{3}4}^{1}$ is non empty (hence $B_{2,\bar{3}\bar{4}}^{1}$
is empty).

\begin{enumerate}
\item If $Q_{2,3\bar{4}}$ is empty (no coding opportunity), node 2 transmits the packet stored
in buffer $B_{2,\bar{3}4}^{1}$.

\begin{enumerate}
\item If the packet is received by node 3, it is removed $B_{2,\bar{3}4}^{1}$.
\end{enumerate}
\item \label{enu:net-code1}If $Q_{2,3\bar{4}}$ is non-empty (coding opportunity), node 2
transmits the network-coded packet $p=q_{1}\oplus q_{2}$, where $q_{1}$
is the packet stored in buffers $B_{2,\bar{3}4}^{1}$ and $B_{4,\bar{3}}^{1}$, and
$q_{2}$
is the packet at the head of queues $Q_{2,3\bar{4}}$ and $Q_{3,\bar{4}}^{2}$.

\begin{enumerate}
\item If packet $p$ is received by node 4, node 4 decodes  packet $q_{2}=p\oplus q_1$ and $q_2$ is removed from
$Q_{2,3\bar{4}}$ and $Q_{3,\bar{4}}^{2}$.
\item If packet $p$ is received by node 3, node 3 decodes packet $q_{1}=p\oplus q_2$ and $q_2$ is removed from
buffers $B_{2,\bar{3}4}^{1}$ and $B_{4,\bar{3}}^{1}$ .
\end{enumerate}
\end{enumerate}
\end{enumerate}
\bibliographystyle{IEEEtran}
\bibliography{References}

\end{document}